\documentclass[11pt]{article}
\usepackage{graphicx,amsthm,graphicx,fancyhdr,mathrsfs}
\usepackage{amsfonts}
\usepackage{amsmath}
\usepackage{amssymb}
\usepackage{url}
  \usepackage[colorlinks=true,citecolor=blue]{hyperref}
\usepackage{fancyhdr}
\usepackage{indentfirst}
\usepackage{enumerate}
\usepackage{amsthm}
\usepackage{color}
\usepackage{comment}
\usepackage{dsfont}
\usepackage{natbib}
\usepackage[misc]{ifsym}
\usepackage[toc,page]{appendix}

\newcommand{\VaR}{\mathrm{VaR}}

\newcommand{\E}{\mathbb{E}}
\newcommand{\R}{\mathbb{R}}

\newcommand{\p}{\mathbb{P}}

\newcommand{\id}{\mathds{1}}

\newcommand{\X}{\mathcal X}

\newcommand{\esssup}{\mathrm{ess\mbox{-}sup}}
\newcommand{\essinf}{\mathrm{ess\mbox{-}inf}}
\renewcommand{\ge}{\geqslant}
\renewcommand{\le}{\leqslant}
\renewcommand{\geq}{\geqslant}
\renewcommand{\leq}{\leqslant}
\renewcommand{\epsilon}{\varepsilon}

\usepackage{array}
\newcommand{\PreserveBackslash}[1]{\let\temp=\\#1\let\\=\temp}
\newcolumntype{C}[1]{>{\PreserveBackslash\centering}p{#1}}
\newcolumntype{R}[1]{>{\PreserveBackslash\raggedleft}p{#1}}
\newcolumntype{L}[1]{>{\PreserveBackslash\raggedright}p{#1}}

\usepackage{booktabs,array}

\newcount\rowc

\makeatletter
\def\ttabular{%
\hbox\bgroup
\let\\\cr
\def\rulea{\ifnum\rowc=\@ne \hrule height 1.3pt \fi}
\def\ruleb{
\ifnum\rowc=1\hrule height 1.3pt \else
\ifnum\rowc=6\hrule height \heavyrulewidth
   \else \hrule height \lightrulewidth\fi\fi}
\valign\bgroup
\global\rowc\@ne
\rulea
\hbox to 10em{\strut \hfill##\hfill}%
\ruleb
&&%
\global\advance\rowc\@ne
\hbox to 10em{\strut\hfill##\hfill}%
\ruleb
\cr}
\def\endttabular{%
\crcr\egroup\egroup}

\theoremstyle{plain}
\newtheorem{theorem}{Theorem}

\newtheorem{lemma}{Lemma}

\theoremstyle{definition}

\newtheorem{example}{Example}

\newtheorem{remark}{Remark}

\usepackage{tikz}

\usepackage[compact]{titlesec}

\makeatletter
\DeclareRobustCommand{\bsquare}{%
  \mathop{\vphantom{\sum}\mathpalette\bigstar@\relax}\slimits@
}

\newcommand{\bigstar@}[2]{%
  \vcenter{%
    \sbox\z@{$#1\sum$}%
    \hbox{\resizebox{.9\dimexpr\ht\z@+\dp\z@}{!}{$\m@th\dsquare$}}%
  }%
}
\makeatother

\newcommand{\dsquare}{\mathop{  \square} \displaylimits}

\usepackage[onehalfspacing]{setspace}

\begin{document}

\title{Pairwise counter-monotonicity}
\author{Jean-Gabriel Lauzier\thanks{Department of Statistics and Actuarial Science, University of Waterloo,  Canada. \Letter~{\url{jlauzier@uwaterloo.ca}}}\and Liyuan Lin\thanks{Department of Statistics and Actuarial Science, University of Waterloo,  Canada. \Letter~{\url{liyuan.lin@uwaterloo.ca}} } \and  Ruodu Wang\thanks{Department of Statistics and Actuarial Science, University of Waterloo,  Canada. \Letter~{\url{wang@uwaterloo.ca}}}}
\maketitle
\begin{abstract}
We systematically study pairwise counter-monotonicity, an extremal notion of negative dependence. 
A stochastic representation and an invariance property are established for this dependence structure. 
We show that pairwise counter-monotonicity implies negative association, and it is equivalent to joint mix dependence if both are possible for the same marginal distributions.  
We find an intimate connection between pairwise counter-monotonicity  and 
risk sharing problems for quantile agents. This result highlights the importance of this extremal negative dependence structure in  optimal allocations for agents who are not risk averse in the classic sense.  
\end{abstract}

\noindent \emph{Declaration of interests:} None.

\section{Introduction}

  Dependence modelling is a crucial part of modern quantitative studies in economics, finance, and insurance  
   (\cite{MFE15}). 
  Comonotonicity and counter-monotonicity
 are known as the strongest forms of positive and negative dependence, respectively.
 In quantitative risk management,   assuming knowledge of the marginal distributions, comonotonicity corresponds to the most dangerous dependence structure (\cite{D94} and \cite{DDGKV02,DVGKTV06}) for the aggregate risk,
 whereas counter-monotonicity 
 corresponds to the safest. 
 In dimensions higher than $2$, 
 by counter-monotonicity we mean pairwise counter-monotonicity (\cite{D72}), which has been studied under the name of mutual exclusivity in the actuarial literature (\cite{DD99} and \cite{CL14}).\footnote{Mutual exclusivity  is defined using joint exceedance probability (see Section \ref{sec:preliminaries}). The two definitions are shown to be equivalent first by \citet[Lemma 2]{D72} and in a more precise form by \citet[Theorem 4.1]{CL14}.}

Despite the obvious similarity in  their definitions, comonotonicity and counter-monotonicity 
 are asymmetric in several major senses. 
 For instance,  comonotonicity admits a stochastic representation (see   Lemma \ref{th:como} below), but such a representation is not known for pairwise counter-monotonicity. 
 Moreover, for any given tuple of marginal distributions, a comonotonic random vector with these marginal distributions always exists, 
 but a pairwise counter-monotonic one may not exist unless quite restrictive conditions on the marginal distributions are satisfied, as first studied by \cite{D72}. In particular, a pairwise counter-monotonic random vector cannot have continuous marginal distributions.
 Comonotonicity 
 has many important roles in   economics, finance and actuarial science, and as such it has received great attention in the literature,
 as in
 axiomatization of preferences (\cite{Y87,S89}), risk measures (\cite{K01}) and premium principles (\cite{WYP97}), risk sharing (\cite{LM94,JST08}),  insurance design (\cite{HMS83, CD03}),  risk aggregation (\cite{EWW15}), and  optimal transport (\cite{R13}).

In sharp contrast to the rich literature on comonotonicity, research on pairwise counter-monotonicity is quite limited.    
As a dependence concept,  
pairwise counter-monotonicity has been studied by \cite{D72},  \cite{HW99}, \cite{DD99} and \cite{CL14}, but the list of relevant studies do not grow much longer.
In contrast to the relatively limited studies on pairwise counter-monotonicity,
this dependence structure appears naturally in many economic contexts, such as lottery tickets, Bitcoin mining, gambling, and  mutual aid platforms, whenever payment events are mutually exclusive.
In particular, 
the interest in studying pairwise counter-monotonicity has grown in the recent risk sharing literature.   
A pairwise counter-monotonic structure is the essential building block of any optimal allocation for  agents using Value-at-Risk  (VaR, which are quantiles) and quantile-related risk measures; such problems are studied by 
\cite{ELW18} and  generalized by 
\cite{W18}, \cite{ELMW20}, \cite{LMWW22} and \cite{XZH23}.  
{Moreover,  counter-monotonicity, when possible,
serves as the best-case dependence structure in risk aggregation for some common risk measures, and, in some contexts, it also serves as the worst-case dependence structure for VaR (see Example \ref{ex:r1-1} in Section \ref{sec:preliminaries}).}

This paper is dedicated to a {systematic study} of pairwise counter-monotonicity. 
As comonotonicity and counter-monotonicity are classic and prominent concepts in mathematics and its applications with a long history, at least since the seminal work of \cite{HLP34}, one may guess that 
there is not much more to discover about  them.
To our pleasant surprise, we offer, through the development of this paper, many new {results on counter-monotonicity, some of which are motivated by recent developments in risk management.}

We obtain a new stochastic representation for pairwise counter-monotonic random vectors using their component-wise sum in Theorem \ref{th:PCT}, which will be useful for many other results in the paper.
The  second result, 
Theorem \ref{th:CT}, establishes that 
counter-monotonicity 
  is preserved under increasing transforms on 
  disjoint sets of components of a random vector,
which is an invariance property proposed by \cite{JP83} satisfied by negative association  (\cite{AS81}). 
Using this invariance property, we obtain in
 Theorem \ref{th:NA} that 
 counter-monotonicity implies negative association.
 The notion of negative association   is stronger than many other forms of negative dependence, such as negative orthant dependence (\cite{BSS82}) and negative supermodular dependence (\cite{H00}).
  In particular,  Theorem \ref{th:NA}  surpasses a result of 
\cite{DD99} showing that  counter-monotonicity implies 
negative supermodular dependence. 

 Another negative dependence concept is joint mix dependence (\cite{WW11,WW16}),  {which can be used to optimize} many quantities in risk aggregation; see \cite{WPY13} and \cite{R13}. 
To connect 
counter-monotonicity and joint mix dependence, we 
fully characterize all Fr\'echet classes (\cite{J97}) 
which are compatible with both dependence concepts in Theorem \ref{th:frechet}; 
it turns out that the two notions, when both exist in the same Fr\'echet class, are equivalent. 
Finally, we show in Theorem \ref{th:quantile} that in the context of risk sharing for quantile agents (\cite{ELW18}), 
under some mild conditions on the total loss,
there always exists a pairwise counter-monotonic Pareto-optimal allocation, and 
any pairwise counter-monotonic allocation is  Pareto optimal for some agents. 
As a consequence, pairwise counter-monotonic random vectors are natural for  agents that are not risk averse.
This is in stark contrast to comonotonic allocations, which appear prominently for risk-averse agents (in the sense of \cite{RS70}) as a consequence of comonotonic improvements introduced by \cite{LM94}.

 \section{Preliminaries}\label{sec:preliminaries}

   We first define comonotonicity and counter-monotonicity for bivariate random variables. 
Fix a  probability space $(\Omega,\mathcal A,\p)$. The probability space does not need to be atomless in Sections \ref{sec:preliminaries}-\ref{sec:4}.
We treat almost surely (a.s.)~equal random variables as identical; this means that 
all equalities and inequality for random variables hold in the a.s.~sense, and we omit ``a.s." in all our statements. Terms like ``increasing" are in the non-strict sense. 
 Let $n$ be a positive integer  and $[n]=\{1,\dots,n\}$. Throughout,  
we consider $n\ge 2$.

  A bivariate random vector $(X,Y)$ is \emph{comonotonic} if there exist increasing functions $f,g$ and a random variable $Z$ such that $(X,Y)=(f(Z),g(Z))$. 
A bivariate random vector $(X,Y)$ is \emph{counter-monotonic} if  $(X,-Y)$ is comonotonic. 
An equivalent formulation of comonotonicity is 
$$
(X(\omega)-X(\omega'))(Y(\omega)-Y(\omega')) \ge 0~~\mbox{for $(\p\times \p)$-almost every~}(\omega,\omega')\in \Omega^2. 
$$
An equivalent formulation of counter-monotonicity is 
$$
(X(\omega)-X(\omega'))(Y(\omega)-Y(\omega')) \le 0~~\mbox{for $(\p\times \p)$-almost every~}(\omega,\omega')\in \Omega^2. 
$$

Next, we define these concepts in dimensions higher than $2$.
For $n\ge 3$, a random vector $\mathbf X$  taking values in $\R^n$ is \emph{(pairwise) comonotonic} if each pair of its components is comonotonic,
and it  is \emph{(pairwise) counter-monotonic} if each pair of its components is counter-monotonic.\footnote{We also say that random variables $X_1,\dots,X_n$ are comonotonic (counter-monotonic), which means that the random vector $(X_1,\dots,X_n)$ is comonotonic (counter-monotonic). 
} 
We will say  ``pairwise counter-monotonicity"  to emphasize the case $n\ge 3$ 
and simply say ``counter-monotonicity" when we also include dimension $2$.
We always omit ``pairwise" for comonotonicity, for which the distinction between dimensions $n=2$ and $n\ge 3$ is unnecessary.

 There are many equivalent ways of formulating comonotonicity and counter-monotonicity; see \citet[Section 3.2]{PW15} for a review.
 For instance, they can be formulated using joint distributions. 
A comonotonic random vector and a counter-monotonic random vector have, respectively, the largest and the smallest joint distribution functions among all random vectors with the same marginals.
With given marginals, the largest (resp.~smallest) joint distribution function is known as the Fr\'echet-Hoeffding upper (resp.~lower) bound. 

A stochastic representation of comonotonicity, which  follows from \citet[Proposition 4.5]{D94}, is presented in the  next lemma.

\begin{lemma}[\cite{D94}]\label{th:como}
Let  $(X_1,\dots,X_n)$ be a random vector and denote by $S=\sum_{i=1}^n X_i$. 
The following are equivalent. 
\begin{enumerate}[(i)]
\item $(X_1,\dots,X_n)$ is comonotonic. 
\item There exist     increasing functions $f_1,\dots,f_n$ and a random variable $Z$  such that 
$
X_i=f_i(Z)$  {for all $i\in [n]$}.
\item There exist continuously increasing functions $f_1,\dots,f_n$ such that 
$
X_i=f_i(S)$ 
{for all $i\in [n]$}.
\end{enumerate}
\end{lemma} 
Lemma \ref{th:como} implies that a comonotonic vector can be represented by increasing functions of the sum $S$. 
Such a representation result does not exist for pairwise counter-monotonicity, since the sum $S$ cannot determine the components $(X_1,\dots,X_n)$ in the presence of negative dependence.

Although quite different from comonotonicity, pairwise counter-monotonicity
also has a special structure, presented below in Lemma \ref{lem:pairwise},
which 
is a restatement of Lemma 2 and Theorem 3 of \cite{D72}.
This result will be useful in a few places in the paper.  
The current form of  this lemma can be found in 
 Theorem 4.1 of \cite{CL14} and Proposition 3.2 of \cite{PW15}.  Denote by $\essinf  X$ and $\esssup X$ the essential infimum and essential supremum of a random variable $X$, respectively.

\begin{lemma}[\cite{D72}] \label{lem:pairwise} 
If at least three of $X_1,\dots,X_n$ are non-degenerate,
pairwise counter-monotonicity of $(X_1,\dots,X_n)$ means that one of the following two cases holds true:
\begin{equation}
\label{eq:PCT1} 
\p(X_i>\essinf X_i,~ X_j>\essinf X_j)=0 \mbox{ for all $i\ne j$};  
\end{equation} 
\begin{equation}
\label{eq:PCT2} 
\p(X_i<\esssup X_i,~ X_j<\esssup X_j)=0 \mbox{ for all $i\ne j$}.
\end{equation}
A necessary condition for \eqref{eq:PCT1} is 
$
\sum_{i=1}^n \p(X_i>\essinf X_i)\le 1,
  $
  and
a necessary condition for \eqref{eq:PCT2} is 
$
\sum_{i=1}^n \p(X_i<\esssup X_i)\le 1.
  $
\end{lemma} 
In the actuarial literature, mutual exclusivity of $(X_1,\dots,X_n)$ is defined as either \eqref{eq:PCT1} or \eqref{eq:PCT2}; see \cite{CL14}.

Pairwise counter-monotonicity imposes strong constraints on the marginal distributions. 
For instance, the necessary condition in case of \eqref{eq:PCT1} is equivalent to $\sum_{i=1}^n \p(X_i=\essinf X_i)\ge n-1$,
and it implies, in particular, that $X_1,\dots,X_n$ are bounded from below.
Moreover, given $n\ge 3$ non-degenerate marginal distributions, a pairwise counter-monotonic random vector exists if and only if one of the two necessary conditions on the marginal distributions holds (Theorem 3 of \cite{D72}).  
\begin{example}\label{ex:r1-1}
{We illustrate the special role of counter-monotonicity in risk aggregation with a  simple model. 
Let  $F_1,\dots,F_n$ be Bernoulli distributions with mean $\epsilon\in (0,1/n)$. These distributions may represent losses from credit default events in a pre-specified period, which usually occur with a small probability.  
In risk aggregation problems (e.g., \cite{EPR13,EWW15}), we are interested in the minimum (best-case) value or maximum (worst-case)  value of 
\begin{align}\label{eq:r1-1}
\rho\left(\sum_{i=1}^n X_i\right) \mbox{~~with the marginal condition $X_i\sim F_i, ~i\in [n]$},
\end{align}
where $\rho$ is a risk measure, and $\sum_{i=1}^n X_i$ represents the total loss from a portfolio of defaultable bonds, with the probability of default $\epsilon$ estimated from the credit rating of these bonds, assumed to be equal for simplicity.
We consider two choices of $\rho$, which lead to opposite conclusions. 
\begin{enumerate}[(a)]
\item
Let $\rho$ be a risk measure that is consistent with convex order. Such risk measures are characterized by \cite{MW20}, and they include all law-invariant coherent, as well as convex, risk measures, such as the Expected Shortfall (\cite{FS16}). 
The \emph{minimum} value of \eqref{eq:r1-1}
is obtained  by a  counter-monotonic random vector $(X_1,\dots,X_n)$. 
This result holds for other marginal distributions as long as a counter-monotonic random vector with these marginal distributions exists; see e.g., Lemma 3.6 of \cite{CL14}.  
\item  Let $\rho:X\mapsto \inf\{x\in \R: \p(X\le x) \ge 1-\alpha\}$, which is the risk measure $\VaR_\alpha$ in Section \ref{sec:6}.
Further, assume that $\alpha/\epsilon  \in (n/2,n)$. The \emph{maximum} value of  \eqref{eq:r1-1} 
is obtained by a  counter-monotonic random vector $(X_1,\dots,X_n)$, as explained below.
First, since  $\sum_{i=1}^n X_i$  only takes integer values, so does $\rho( \sum_{i=1}^n X_i)$. 
If $(X_1,\dots,X_n)$ is counter-monotonic, then $\sum_{i=1}^n X_i$ follows a Bernoulli distribution with mean $n\epsilon>\alpha$, and hence $\rho(\sum_{i=1}^n X_i)=1$. Moreover, for any $X_1,\dots,X_n$ with the specified marginal distributions, if $\rho(\sum_{i=1}^n X_i)\ge 2$ then $\E[\sum_{i=1}^n X_i]\ge 2\alpha> n\epsilon$, a contradiction, thus showing $\rho(\sum_{i=1}^n X_i)\le 1$.   
\end{enumerate}
The interpretation of the above two cases is that, for credit default losses, using a coherent risk measure  
and using VaR may lead to opposite conclusions on which dependence structure is safe or dangerous, and both cases highlight the important role of counter-monotonicity.}
\end{example}

\section{Stochastic representation of pairwise counter-monotonicity}

We provide in this section a stochastic representation of pairwise counter-monotonicity.  
To explain the result, let $\Pi_n$ be the set of all $n$-compositions of $\Omega$, that is, $$\Pi_n=\left\{(A_1,\dots,A_n)\in \mathcal A^n: \bigcup_{i\in [n]} A_i=\Omega \mbox{~and~$A_1,\dots,A_n$ are disjoint}\right\}.$$
In other words, a composition of $\Omega$ is a partition of $\Omega$ with order.
Denote by $\mathcal X_\pm$ the set of  all nonnegative random variables and nonpositive random variables.

 \begin{theorem}\label{th:PCT}
Let  $(X_1,\dots,X_n)$ be a random vector and denote by $S=\sum_{i=1}^n X_i$. 
Suppose that at least three of $X_1,\dots,X_n$ are non-degenerate.  
The following are equivalent. 
\begin{enumerate}[(i)]
\item $(X_1,\dots,X_n)$ is pairwise counter-monotonic.
\item There exist 
$ m_1,\dots,m_n\in \R$,  $(A_1,\dots,A_n)\in\Pi_n$ and  $Z\in \mathcal X_\pm$ such that 
\begin{align}\label{eq:CT1}
X_i =  Z  \id_{A_i} +m_i~~~\mbox{for all $i\in [n]$}. \end{align}   
\item   There exists $(A_1,\dots,A_n)\in\Pi_n$ such that   
 \begin{align}\label{eq:END}
X_i =  (S-m)  \id_{A_i} +m_i~~~\mbox{for all $i\in [n]$},
 \end{align} 
 where either  $m_i=\essinf X_i$ for $i\in [n]$ or $m_i=\esssup X_i$ for $i\in [n]$, and $m=\sum_{i=1}^n m_i$.
\end{enumerate}
\end{theorem}
\begin{proof}
The implication (iii)$\Rightarrow$(ii) is straightforward.
To see (ii)$\Rightarrow$(i),  take $i,j\in [n]$ with $i\ne j$, and we check a few cases of $\omega,\omega'\in \Omega$. 
  If   $\omega,\omega' \not\in A_i $, then    $X_i(\omega)=X_i(\omega')=m_i$,  
and hence
\begin{align}
\label{eq:check}
(X_i(\omega) -X_i(\omega'))(X_j(\omega) -X_j(\omega'))=0.
\end{align}
Similarly, \eqref{eq:check} holds  if $\omega,\omega' \not\in A_j $. 
  If $(\omega,\omega')\in A_i \times A_j$ or $(\omega,\omega')\in A_j \times A_i$, then 
$$
(X_i(\omega) -X_i(\omega'))(X_j(\omega) -X_j(\omega'))=-Z(\omega)Z(\omega') \le0.
$$  
This shows that $(X_i,X_j)$ is counter-monotonic, and hence, $(X_1,\dots,X_n)$ is pairwise counter-monotonic. 

Next,  we show the implication (i)$\Rightarrow$(iii). 
By Lemma \ref{lem:pairwise}, it suffices to consider   \eqref{eq:PCT1} and \eqref{eq:PCT2}. 
Suppose that  \eqref{eq:PCT1} holds. 
Let $B_i=\{X_i>\essinf X_i\}$ and $m_i=\essinf X_i$
 for $i\in [n]$. 
Clearly $B_1,\dots,B_n$ are  (a.s.)~disjoint events, and $S\ge  \sum_{i=1}^n m_i= m$.
Using \eqref{eq:PCT1},  if   event $B_i$ occurs,
then  $X_j=m_j$ for $j\ne i$, and
$S=X_i+\sum_{j=1}^n m_j -m_i$.
Moreover, if $B_i$ does not occur, then 
$X_i=m_i$. 
Therefore, we have 
\begin{equation}\label{eq:construct1}
X_i = (S-m+m_i) \id_{B_i}  +m_i\id_{B^c_i} = (S-m)\id_{B_i} +m_i     ,~~\mbox{for $i\in [n]$}.
\end{equation} 
Let $B= \{S=m\} $ and it is clear that $(B,B_1,\dots,B_n)$  is a  composition of $\Omega$.
Let $A_1=B_1\cup  B$, and $A_2=B_2,\dots,A_n=B_n$. 
Since $S-m=0$ on $B$,  \eqref{eq:construct1} yields \eqref{eq:END}.
If  \eqref{eq:PCT2} holds instead of \eqref{eq:PCT1}, then we can analogously show  \eqref{eq:END} with $m_i=\esssup X_i$ for $i\in [n]$. 
\end{proof}
 
   Theorem \ref{th:PCT} shows that 
pairwise counter-monotonicity can be represented by   the sum $S$ and a composition $(A_1,\dots,A_n)$. 
In contrast, comonotonicity can be represented by the sum $S$ and increasing continuous functions $f_1,\dots,f_n$, as in  Lemma \ref{th:como}.
This representation result will be instrumental  in proving the other results of this paper. 
Another direct consequence of Theorem \ref{th:PCT} is that  if at least three components of a   pairwise counter-monotonic random vector are  non-degenerate, then  either the components are all  bounded from below or they are all  bounded from above; this can also be seen from Lemma \ref{lem:pairwise}. 

\begin{example}\label{ex:simple}
A simple pairwise counter-monotonic random vector in the form of \eqref{eq:CT1} and \eqref{eq:END}, which will be referred to repeatedly in the following sections, is given by
\begin{align}
\label{eq:ex-simple}
X_i=\id_{A_i} \mbox{~for $i\in [n]$ where $(A_1,\dots,A_n)\in \Pi_n$.} 
\end{align}
Such $(X_1,\dots,X_n)$ may represent the outcome of $n$ lottery tickets, exactly one of which randomly wins a reward of $1$, or the reward to Bitcoin miners computing the next block in the Bitcoin blockchain; see \cite{LS20}.  
\end{example}

\begin{remark}
In parts (ii) and (iii) of Theorem \ref{th:PCT}, 
we can replace $(A_1,\dots,A_n)\in \Pi_n$ by   $A_1,\dots,A_n$ being disjoint events, and 
the equivalence relations in the theorem remain true.
\end{remark}

In the case  
 at most two components of $(X_1, \dots, X_n)$ are non-degenerate,
 the stochastic representation of counter-monotonicity is quite different from Theorem \ref{th:PCT}. 
When $n=2$,
  $(X_1,X_2)$ is counter-monotonic if and only if 
 there exist increasing functions $f_1,f_2$ such that 
$$
X_1=f_1(X_1-X_2) ~~~\mbox{and}~~~ X_2=f_2(X_2-X_1);
 $$
this statement follows by applying Lemma \ref{th:como} to the comonotonic random vector $(X_1,-X_2)$. 
Note that the difference $X_1-X_2$ 
replaces the summation $S=X_1+X_2$ in Lemma \ref{th:como}.
The sum of two 
counter-monotonic random variables represents the loss from a hedged portfolio and  it has been studied by \cite{CDLT14} and \cite{CCGM20}.




%
%
%
\section{Invariance property and negative association}
\label{sec:4}


Negative association appears in various natural probabilistic and statistical contexts, such as permutation distributions, sampling without replacement, negatively correlated Gaussian distributions and tournament scores; see \cite{JP83} and the more recent paper \cite{CRW22} for many examples. 
 
A random vector $\mathbf X= (X_1,\dots,X_n)$  is said to be \emph{negatively associated} if  for any disjoint subsets $I,J \subseteq [n]$, and any real-valued, coordinate-wise increasing functions $f,g$,  
we have 
\begin{align}\label{eq:def-na}
\text{Cov}(f(\mathbf X_I),g(\mathbf X_J)) \leq 0,
\end{align}
where $\mathbf X_I=(X_k)_{k\in I}$ and $\mathbf X_J=(X_k)_{k\in J}$,
provided that $f(\mathbf X_I)$ and $g(\mathbf X_J)$ have finite second moments.  
Negative association  is stronger than many other notions of negative dependence, such as negative supermodular dependence (shown by \cite{CV04}) and negative orthant dependence (shown by \cite{JP83}).
\begin{remark}
Negative association is invariant under increasing marginal transforms. Therefore, if $f(\mathbf X_I)$ and $g(\mathbf X_J)$ are continuously distributed, then NA implies that \eqref{eq:def-na} holds with the covariance operator replaced by  Spearman's rank correlation coefficient or another similar concordance measure; see \citet[Chapter 7]{MFE15}. 
\end{remark}

We first present a self-consistency property of both comonotonicity and counter-monotonicity in the spirit of Property P6 of \cite{JP83} for negative association. 
 To the best of our knowledge, this self-consistency property is not found in the literature even for the case of comonotonicity, although its proof is straightforward. 
\begin{theorem}\label{th:CT}
The following statements hold.
\begin{enumerate}[(i)]
\item 
Increasing functions  of subsets of a set of comonotonic random variables
are comonotonic.
\item 
Increasing functions of  disjoint subsets of a set of counter-monotonic random variables
are counter-monotonic.
\end{enumerate}

\end{theorem} 

\begin{proof} 
(i) Let $\mathbf X=(X_1, \dots, X_n)$ be a comonotonic random vector. By Lemma \ref{th:como},  there exist     increasing functions $f_1,\dots,f_n$ and a random variable $Z$  such that $X_i=f_i(Z)$
for all $i\in [n]$.
For $I_1, \dots, I_m \subseteq [n]$
and  increasing functions $g_j:\R^{\vert I_j\vert} \to \R$,  $j\in [m]$, 
let 
$Y_j=g_j(\mathbf X_{I_j})$,  $j\in [m]$, where $|\cdot |$ is the cardinality of a set. 
That is, $Y_j=g_j\circ f_{I_j}(Z)$ where $f_{I_j}=(f_i)_{i\in I_j}$.
As the composition of increasing functions, $g_i\circ f_{I_j}$ is increasing on $\R$. Thus, $(Y_1, \dots, Y_m)$ is a comonotonic vector.

(ii) Let $\mathbf X=(X_1, \dots, X_n)$ be a pairwise counter-monotonic random vector. If at most two of $X_1, \dots, X_n$ are non-degenerate,  the desired statement holds trivially. Next, we assume that  at least three of $X_1, \dots, X_n$ are non-degenerate. 
For disjoint subsets $I_1, \dots, I_m$  of $[n]$ and increasing functions $g_j:\R^{\vert I_j\vert} \to \R$,  $j\in [m]$, let 
 $Y_j=g_j(\mathbf X_{I_j})$,  $j\in [m]$. 
 By Theorem \ref{th:PCT},  there exist 
$ \mathbf m=(m_1,\dots,m_n)\in \R^n$, $(A_1,\dots,A_n)\in\Pi_n$ and  $Z\in \mathcal X_\pm$ such that 
$
X_i =  Z  \id_{A_i} +m_i$ for all $i\in [n]$. 
Without loss of generality, assume $Z\ge 0$. 
For $i\in [n]$ and $j\in [m]$, if $A_i$ occurs, 
then $X_i=Z+m_i$ and $X_k=m_k$ for $k \neq i$, 
which means $Y_j=  g_j (\mathbf X_{I_j}) \ge g_j (\mathbf m_{I_j})$.
If $A_i$ does not occur, then $Y_j=  g_j (\mathbf m_{I_j})$.
Let $Z_j=  \sum_{i\in I_j} \left(g_j(\mathbf{X}_{I_j})-g_j(\mathbf{m}_{I_j})\right)\id_{A_i} \ge 0$. It follows that
\begin{align*}
Y_j &= \sum_{i\in I_j} g_j(\mathbf{X}_{I_j})  \id_{A_i} + g_j(\mathbf{m}_{I_j})
\left(1-\sum_{i\in I_j} \id_{ A_i}\right)
\\&= Z_j \id_{\bigcup_{i\in I_j} A_i}+g_j(\mathbf{m}_{I_j})
 =\left( \sum_{k=1}^m Z_k \right)\id_{\bigcup_{i\in I_j} A_i}+g_j(\mathbf{m}_{I_j}).
\end{align*}
By using Theorem \ref{th:PCT} and the fact that $\sum_{k=1}^m Z_k \ge 0$, we conclude that $(Y_1, \dots, Y_m)$ is pairwise counter-monotonic.
\end{proof}
\begin{remark}
For Theorem \ref{th:CT} (i), an equivalent statement is that 
increasing functions of comonotonic random variables are comonotonic. This is because one can   choose the subsets as $[n]$ and take functions on $\R^n$ which are constant in some dimensions. 
We use the current presentation of statement (i) to show a contrast to  statement (ii). 
\end{remark}

 What we will use from Theorem \ref{th:CT} is the second statement, which leads to 
the next result in this section; that is,  
counter-monotonicity implies negative association. 
Since negative association is stronger than  negative supermodular dependence, 
this result  surpasses Theorem 12 of 
\cite{DD99}, which states that counter-monotonicity is stronger than negative supermodular dependence.
 
\begin{theorem}\label{th:NA}Counter-monotonicity implies negative association.
\end{theorem} 
\begin{proof} 
Let  $\mathbf X$ be an $n$-dimensional   counter-monotonic random vector. Take disjoint subsets $I,J \subseteq [n]$ and    coordinate-wise increasing functions $f:\R^{|I|}\to \R$ and $g:\R^{|J|}\to \R$, where $|\cdot |$ is the cardinality of a set. By Theorem \ref{th:CT} (ii),    $f(\mathbf X_I)$ and $g(\mathbf X_J)$ are counter-monotonic.  The Fr\'echet-Hoeffding inequality (see e.g., Corollary 3.28 of \cite{R13}) yields
$\E[f(\mathbf X_I)g(\mathbf X_J)]\le \E[f(\mathbf X_I)]\E[g(\mathbf X_J)]$ provided that the expectations exist.
Hence, $\mathbf X$ is negatively associated. 
\end{proof}

\citet[Theorem 2.11]{JP83} already noted that the lottery-type random vector \eqref{eq:ex-simple} in  Example \ref{ex:simple} is negatively associated.

The result in Theorem \ref{th:NA} 
has a straightforward interpretation,  
as counter-monotonicity is the extreme form of negative dependence, which intuitively should imply other notions of negative dependence, among which negative association is considered a strong notion; see 
\cite{ASB13} for a comparison of several notions of negative dependence. 

Counter-monotonicity is also stronger than several other notions of negative dependence which are not implied by negative association. 
These notions include  
conditional decreasing in sequence and 
negative dependence in sequence (see \citet[Remark 2.16]{JP83}) and negative dependence through stochastic ordering (see \cite{BSS85}).
These implications can be checked directly with Theorem \ref{th:CT}, thus highlighting its usefulness.

\begin{remark}
A random vector $\mathbf X$  is positively associated if   $
\text{Cov}(f(\mathbf X),g(\mathbf X)) \ge  0$  for all real-valued, coordinate-wise increasing functions $f,g$ (\cite{EPW67}). 
Comonotonicity implies positive association because
$(f(\mathbf X),g(\mathbf X))$ is comonotonic 
by Theorem \ref{th:CT},
and the covariance of a comonotonic pair of  random variables is non-negative due to the Fr\'echet-Hoeffding inequality.
\end{remark}

\section{Joint mix dependence and Fr\'echet classes}
\label{sec:5}

Another type of extremal negative dependence structure 
is the notion of joint mixes. 
In this section, we study the connection between counter-monotonicity and joint mix dependence. 

From now on, assume that 
 the probability space $(\Omega,\mathcal A,\p) $ is atomless. 
A random vector $(X_1, \dots, X_n)$
 is a \emph{joint mix} if $\sum_{i=1}^n X_i$ is a constant   $c$, and in this case
 we say that \emph{joint mix dependence} holds for $(X_1,\dots,X_n)$.
 The constant $c$ is called the center of $(X_1, \dots, X_n)$, and it is obvious that $c=\sum_{i=1}^n \E[X_i] $ if the expectations of $X_1,\dots, X_n$ are finite.
Joint mix dependence is regarded as a concept of extremal negative dependence due to its opposite role to comonotonicity in risk aggregation problems;  see \cite{PW15} and \cite{WW16}.  

The lottery-type random vector in Example \ref{ex:simple} satisfies  both counter-monotonicity and joint mix dependence.
 In the case $n=2$,  joint mix dependence is strictly stronger than  counter-monotonicity.  This result cannot be extended to $n\ge 3$. For example, $(X,X,-2X)$ is a joint mix that is not counter-monotonic.  
 A weaker notion than joint mix dependence 
 is proposed by \cite{LA14}, which does not imply, and is not implied by, counter-monotonicity in dimension $n\ge 3$.  
 
Joint mix dependence and counter-monotonicity share some similarities. 
First, for a  random vector  $(X_1, \dots, X_n)$ with its sum $S=X_1+\dots+X_n$,
if either pairwise counter-monotonicity  or joint mix dependence holds, 
then $X_i$ and $S-X_i$ are counter-monotonic for each $i\in [n]$.
The case of pairwise counter-monotonicity is verified by Theorem \ref{th:CT}, and
the case of joint mix dependence is verified by definition. 
Second, both dependence notions impose
strong conditions on the marginal distributions. The condition for pairwise counter-monotonicity is given in Lemma \ref{lem:pairwise},
and that for joint mix dependence is much more sophisticated; see \cite{WW16} for some sufficient conditions as well as necessary ones. 
This is in contrast to concepts such as comonotonicity, independence, and negative association, 
for which the existence of the corresponding random vectors is always guaranteed for any given marginal distribution. 
Both  joint mix dependence  and counter-monotonicity are used in the tail region to obtain lower bounds   for risk aggregation with given marginal distributions,
 as studied by \cite{BJW14} and  \cite{CDD17}, respectively.

 The next result characterizes marginal distributions that are compatible with both counter-monotonicity and 
joint mix dependence. 
For this, we 
need some notation and terminology.  In what follows, we will use distribution functions to represent distributions.
For an $n$-tuple $(F_1,\dots,F_n)$ of  distributions on $\R$,
a \emph{Fr\'echet class} (see \citet[Chapter 3]{J97}) is defined as 
$$\mathcal{F}_n(F_1, \dots, F_n)=\{\mbox{distribution of }(X_1, \dots, X_n): X_i\sim F_i, ~i\in[ n]\}.$$
We say that a  Fr\'echet class $\mathcal{F}_n(F_1, \dots, F_n)$  supports counter-monotonicity (resp.~joint mix dependence) if  there exists a counter-monotonic random vector (resp.~a joint mix) whose distribution is in this class.
Let $\delta _x$ be the distribution function
of a point-mass at $x\in \R$, and denote by $\Theta_n$ the 
standard $n$-simplex, that is, $\Theta_n=\{(p_1,\dots,p_n)\in [0,1]^n: \sum_{i=1}^n p_i=1\}.$ 
Two distributions $F$ and $G$ are symmetric if $F(x)=1-G(c-x)$, $x\in \R$ for some $c\in \R$. In other words, if $X$ has distribution $F$, then $c-X$ has distribution $G$.

It turns out that all Fr\'echet classes $\mathcal F_n(F_1,\dots,F_n)$ which support both counter-monotonicity and joint mix dependence 
can be characterized explicitly. 
If at least three of $F_1,\dots,F_n$  are non-degenerate, then   $F_1,\dots,F_n$ are two-point distributions given by 
\begin{align}\label{eq:dis}
F_i =p_i \delta _{a+m_i} +(1-p_i) \delta_{m_i} \mbox{~for $i\in [n]$, where $ a, m_1, \dots, m_n \in \R$ and $(p_1,\dots,p_n)\in \Theta_n$.}
\end{align}
If at most two of $F_1,\dots,F_n$  are non-degenerate, then
\begin{equation}
\label{eq:degenerate}\mbox{$F_i$ and $F_j$ are symmetric for some $i,j\in [n]$,
and $F_k$ is degenerate for all $k\in [n]\setminus\{i,j\}$}.
\end{equation}

\begin{theorem}\label{th:frechet}
 A Fr\'echet class  supports both counter-monotonicity and joint mix dependence if and only if 
one of \eqref{eq:dis} and \eqref{eq:degenerate} holds.  
In case both are supported,  
counter-monotonicity and joint mix dependence are equivalent for this Fr\'echet class.
\end{theorem}

%
%
%

\begin{proof} 
We first prove the equivalence statement in the last part of the theorem. 
  Suppose that the Fr\'echet class $\mathcal{F}_n(F_1, \dots, F_n)$ supports  both counter-monotonicity and joint mix dependence. 
\citet[Theorem 3.8]{PW15} shows that if a Fr\'echet class supports  a counter-monotonic random vector, then  a random vector is   counter-monotonic if and only if it is $\Sigma$-counter-monotonic, and moreover, a joint mix is always $\Sigma$-counter-monotonic.
Using these two facts, a joint mix is counter-monotonic for this Fr\'echet class. 
For the conserve statement, 
note that in  $\mathcal{F}_n(F_1, \dots, F_n)$
there exists a unique distribution function 
$$F(x_1,\dots,x_n)= \left(\sum_{i=1}^n F_i(x_i) - d + 1\right)_+, ~~~~(x_1,\dots,x_n)\in \R^n
$$ of a counter-monotonic random vector (Theorem 3.3 of \cite{PW15}).
Since a joint mix with marginal distributions $F_1, \dots, F_n$ is counter-monotonic, its distribution must coincide with $F$. 
This shows that $F$ is the distribution of a joint mix. 
 
Next, we prove the first part of the theorem.
For the ``if" statement,
assume that a Fr\'echet class   $\mathcal{F}_n(F_1, \dots, F_n)$ supports both counter-monotonicity and joint mix dependence.  By the above argument, $\mathcal{F}_n(F_1, \dots, F_n)$ supports a pairwise counter-monotonic joint mix $(X_1, \dots, X_n)$.   First, consider the case that at least three of $F_1,\dots,F_n$ are non-degenerate. Using \eqref{eq:END},
$$X_i=(c-m) \id_{A_i}+m_i, ~~~\mbox{for}~~~ i\in [n],$$
where $(A_1, \dots, A_n) \in\Pi_n$, $c$ is the center of $(X_1, \dots, X_n)$, either $m_i=\essinf(X_i)$ for all $i\in [n]$ or $m_i=\esssup(X_i)$ for all $i\in [n]$, and $m=\sum_{i=1}^n m_i$. It is clear that  $F_i$ has the form   \eqref{eq:dis} by setting $a=c-m$. 
If at most two of $F_1,\dots,F_n$ are degenerate, say $F_i$ and $F_j$, then a joint mix $(X_1,\dots,X_n)$ with marginal distributions $F_1,\dots,F_n$ satisfies $X_i=c-X_j$ for some $c\in \R$, and $X_k$ is a constant for each $k\in [n]\setminus \{i,j\}$.
This implies \eqref{eq:degenerate}.

Finally, we verify the converse statement. 
If $(F_1, \dots, F_n)$ has the form   \eqref{eq:dis}, then take $X_i=a\id_{A_i}+m_i$ with $(A_1, \dots, A_n) \in \Pi_n$ satisfying $\p(A_i)=p_i$ for $i\in[n]$, and  we have $(X_1, \dots, X_n)$ is counter-monotonic by Theorem \ref{th:PCT} and $\sum_{i=1}^n X_i=a+\sum_{i=1}^n m_i$.
If  $(F_1, \dots, F_n)$ has the form   \eqref{eq:degenerate},
then  by taking $X_i $ with distribution $F_i$, $X_j=c-X_i$ with distribution $F_j$ and $c\in \R$, and $X_k $ with distribution $F_k$ for each $k\in [n]\setminus \{i,j\}$, we can directly verify  that $(X_1,\dots,X_n)$ is a counter-monotonic joint mix.
\end{proof}
 
 From the proof of Theorem \ref{th:frechet} (ii), if at least three components of a pairwise counter-monotonic joint mix $\mathbf X=(X_1, \dots, X_n)$ are non-degenerate,  then it has the form $$X_i=a \id_{A_i}+m_i, ~~~\mbox{for}~~~ i\in [n]$$
where $(A_1, \dots, A_n) \in\Pi_n$, $a\in \R$ and $\mathbf m=( m_1, \dots, m_n )\in \R^n$.
If $a\ne 0$, then the random vector $(\mathbf X-\mathbf m)/a$ has a categorical distribution with
$n$ categories and probability vector $(\p(A_1),\dots,\p(A_n))$.

\begin{remark}
Theorem \ref{th:frechet} characterizes a  Fr\'echet class that supports both counter-monotonicity and joint mix dependence. 
 Fr\'echet classes that support (non-degenerate) pairwise counter-monotonicity are fully described by the conditions in Lemma \ref{lem:pairwise}. 
Whether a given Fr\'echet class supports joint mix dependence is a very challenging problem, with existing result summarized in \cite{PW15} and \cite{WW16}.
In risk aggregation problems, the
notion of joint mix dependence is more relevant, because  a joint mix usually
``approximately exists" for large dimensions, which leads to the main idea behind the Rearrangement Algorithm;   see \cite{EPR13,EPRWB14}, \cite{BV15} and \cite{BRV17}. In contrast,   counter-monotonicity is more relevant for risk sharing problems, which we discuss in the next section.
\end{remark}

\section{Optimal allocations in risk sharing for quantile agents}
\label{sec:6}

We now formally establish the link between counter-monotonicity and Pareto-optimal allocations in risk sharing problems  for quantile agents.

We first describe the basic setting. 
A quantile agent assesses risk by its quantile, also known as the risk measure Value-at-Risk (VaR) in risk management. 
Following the convention of \cite{ELW18}, the VaR at level $\alpha \in (0,1)$ is defined as$$
	\VaR_\alpha(X)=\inf\{x\in \R: \p(X\le x) \ge 1-\alpha\},~~~X\in \mathcal X,
	$$ 
	where $\mathcal X$ is the set of all random variables in the probability space.
	Moreover, write $\VaR_\alpha=-\infty$ on $\mathcal X$ for $\alpha \ge 1$, although our agents use $\VaR_\alpha$ for $\alpha \in(0,1)$.
It is important to highlight that quantile agents  with level $\alpha\in (0,1)$ are not risk averse (\cite{RS70}). 
	
We consider the risk sharing problem for $n\ge 3$ quantile agents
with levels $\alpha_1,\dots,\alpha\in (0,1)$. For a given $S\in \mathcal X$,  the set of \emph{allocations} of $S$ is 
 $$ 
\mathbb{A}_n(S)=\left\{(X_1,\ldots,X_n)\in \mathcal X^n: \sum_{i=1}^nX_i=S\right\}.   $$
 An allocation $(X_1,\dots,X_n)\in \mathbb{A}_n(S)$  is     \emph{Pareto optimal} if for any
$(Y_1,\dots,Y_n)\in \mathbb{A}_n(S)$ satisfying $\VaR_{\alpha_i}(Y_i)\leq\VaR_{\alpha_i}(X_i)$ for all $i\in [n]$, we have $\VaR_{\alpha_i}(Y_i)=\VaR_{\alpha_i}(X_i)$ for all $i\in [n]$.   Pareto optimality of $(X_1,\dots,X_n) \in \mathbb{A}_n(S)$ is equivalent to 
\begin{align}
\label{eq:optimality}
 \sum_{{i=1}}^n\VaR_{\alpha_i}(X_i)=\inf\left\{\sum_{i=1}^n\VaR_{\alpha_i}(Y_i):  (Y_1, \dots, Y_n) \in \mathbb{A}_n(S)\right\}=\VaR_{\sum_{i=1}^n \alpha_i}(S),
\end{align}
where the first equality is \citet[Proposition 1]{ELW18} and the second equality is \citet[Corollary 2]{ELW18}. Using \eqref{eq:optimality}, we obtain that the existence of a Pareto-optimal allocation is equivalent to  $ \sum_{i=1}^n \alpha_i<1$; this is also given by Theorem 3.6 of \cite{WW20}.  
  For this reason, we say that the $n$ quantile agents are \emph{compatible}  if  $ \sum_{i=1}^n \alpha_i<1$  holds, meaning that a Pareto-optimal allocation exists for some $S$, and equivalently, for every $S$.  

 The following theorem shows that, under some conditions of the total risk $S$ to share,  the 
risk sharing problem for any quantile agents admits a  pairwise counter-monotonic Pareto-optimal allocation, and 
every pairwise counter-monotonic allocation is Pareto optimal for some agents.
Moreover, comonotonic allocations are never Pareto optimal. 
Recall that by Lemma \ref{lem:pairwise}, 
a pairwise counter-monotonic random vector  $ (X_1,\dots,X_n)$ satisfies either 
\eqref{eq:PCT1} or \eqref{eq:PCT2}.

\begin{theorem}\label{th:quantile}
For $S\in \X$, the following hold.   
\begin{enumerate}[(i)]
\item  If $S$ is bounded from below, then for any compatible quantile agents
there exists a pairwise counter-monotonic allocation of $S$ which is Pareto optimal.
\item If $\p(S=\essinf S)>0$, then every type-\eqref{eq:PCT1} pairwise counter-monotonic allocation of $S$   is Pareto optimal for some quantile agents.
\item If $S$ is continuously distributed, then a comonotonic allocation of $S$ is never Pareto optimal for any quantile agents.

\end{enumerate} 
\end{theorem} 

\begin{proof} 
(i)  Let $\alpha_1,\dots,\alpha_n\in (0,1)$ be the VaR levels of the quantile agents.
Compatibility  of the agents means $\sum_{i=1}^n \alpha_i<1$.
In this case, a Pareto-optimal allocation $(X_1,\dots,X_n)$ of $S$ is given by Theorem 2 of \cite{ELW18}, with the form
$$
X_i= (X-m) \id_{A_i} ,~ i \in[n-1]~~~\mbox{and}  ~~~X_n = (X-m)\id_{A_n}+m
$$
for some  $(A_1,\dots,A_n)\in \Pi_n$.  By setting $m=\essinf S$,  $(X_1,\dots,X_n)$ is pairwise counter-monotonic by Theorem \ref{th:PCT}.  

(ii)  Note that shifting $X_1,\dots,X_n$ by arbitrary constants, and adjusting $S$ correspondingly, does not affect its Pareto optimality due to \eqref{eq:optimality}.
Moreover, \eqref{eq:PCT1} guarantees that at  most one of $X_1,\dots,X_n$ is not bounded from below, and further $\p(S=\essinf S)>0$ guarantees that 
this can only happen if all $X_1,\dots,X_n$ are bounded from below. 
Therefore, we can, 
  without loss of generality,   assume $\essinf X_i=0$ for each $i\in [n]$. 
  
Let $B=\{S=\essinf S\}$   and $A=\bigcup_{i=1}^n \{X_i> 0 \}$. 
First, if $\p(B\cap A) =0$, then we let $\alpha_i= \p(X_i>0)+\p(B)/(2n)>0$ for $i\in [n]$. Note that $$\sum_{i=1}^n \alpha_i = \sum_{i=1}^n \p(X_i>0) +\frac 12  \p(B)=\p(A) + \frac 12\p(B) <\p(A)+\p(B)= \p(A\cup B)\le 1.$$
It is clear that $\VaR_{\alpha_i}(X_i)=0$ for each $i\in [n]$, leading to $\sum_{i=1}^n \VaR_{\alpha_i}(X_i) =0 \le  \essinf S \le \VaR_{\sum_{i=1}^n \alpha_i} (S)$.
Note that 
\begin{align}
\label{eq:optimality-2}
\sum_{i=1}^n \VaR_{\alpha_i}(X_i) \le \VaR_{\sum_{i=1}^n \alpha_i} (S)
~\Longrightarrow~ \mbox{$(X_1,\dots,X_n)$ is Pareto optimal}.
\end{align}
This is because  Corollary 1 of \cite{ELW18} gives 
 $ \sum_{i=1}^n \VaR_{\alpha_{i}}(X_i)
 \ge  \VaR_{\sum_{i=1}^n \alpha_i} (S),$
and this leads to $\sum_{i=1}^n \VaR_{\alpha_{i}}(X_i)
 =  \VaR_{\sum_{i=1}^n \alpha_i} (S)$ in \eqref{eq:optimality}, which 
 gives Pareto optimality of $(X_1,\dots,X_n)$ as we see in part (i).

Next, assume $\p(B\cap A)>0$. Then there exists $j\in [n]$ such that 
$\p(B\cap \{X_j>0 \})>0$. Let $\epsilon =\p(B\cap \{X_j>0\})/(2n)$. 
Take $\alpha_i =  \p( X_i> 0 ) +\epsilon >0$ for $i\in[n ]\setminus \{j\}$
and
$\alpha_j= \p(\{X_j>0\}\setminus B) +\epsilon$. 
By Lemma \ref{lem:pairwise}, 
$$
1\ge \sum_{i=1}^n \p(X_i> 0 ) =
 \sum_{i= 1}^n(\alpha_i-\epsilon)  + \p(B\cap \{X_j>0\})  = \sum_{i=1}^n \alpha_i + n \epsilon,
$$
and hence $\sum_{i=1}^n\alpha_i<1$.
By definition of $\alpha_1,\dots,\alpha_n$, we have 
$\VaR_{\alpha_i}(X_i)= 0$ for $i\in[n ]\setminus \{j\}$.
Moreover, note that $X_j=S$ on $\{X_j> 0\}$ and
$$\p(\{X_j= \essinf S\}\cap \{X_j>0\})=\p(B\cap \{X_j>0\}) =2n \epsilon,$$
which implies $\p(X_j> \essinf S)  = \p(X_j>0) - 2n\epsilon  < \alpha_j. $
Therefore, 
$\VaR_{\alpha_j}(X_j)\le \essinf S $, 
leading to $\sum_{i=1}^n \VaR_{\alpha_i}(X_i) \le \essinf S \le \VaR_{\sum_{i=1}^n \alpha_i} (S)$. Hence, we obtain Pareto optimality of $(X_1,\dots,X_n)$  via \eqref{eq:optimality-2}.

(iii) For a comonotonic allocation $(X_1, \dots, X_n)$ of $S$, 
using decreasing monotonicity of $\alpha \mapsto \VaR_\alpha$ and comonotonic additivity of $\VaR_\alpha$, 
we have 
\begin{align*}
\sum_{i=1}^n \VaR_{\alpha_i}(X_i)\ge \sum_{i=1}^n \VaR_{\beta }  (X_i)  = \VaR_{\beta } (S),
\end{align*} 
where we write $\beta =\max\{\alpha_1,\dots, \alpha_n\}$.
As $S$ is continuously distributed, $\VaR_\alpha(S)$ is strictly decreasing in $\alpha$. Noting that $\beta< \sum_{i=1}^n \alpha_i$, we have 
$\VaR_\beta  (S)>\VaR_{\sum_{i=1}^n \alpha_i}(S)$. Therefore, the comonotonic allocation $(X_1, \dots, X_n)$ is not Pareto optimal by \eqref{eq:optimality}. 
\end{proof}
 
Theorem \ref{th:quantile} states that allocations with a pairwise counter-monotonic structure solve the problem of sharing risk among quantile agents. 
For instance, the lottery-type allocation in Example \ref{ex:simple} is Pareto optimal for some quantile agents. 
Further, Theorem \ref{th:quantile} (iii) states that comonotonic allocations can never be Pareto optimal for quantile agents if the total risk is continuously distributed. As mentioned, this is in stark contrast with the risk sharing problem with risk-averse agents, for which comonotonic allocations are always optimal. The latter result, due to the comonotonic improvements of \cite{LM94},  is well-known; see also \cite{JST08} and \cite{R13}. Moreover, when all agents are strictly risk averse, only comonotonic allocations are Pareto optimal  (see \citet[Proposition 4]{LLW23} for the case when preferences are modelled by strictly concave distortion functions).





As a symmetric statement to Theorem \ref{th:quantile}, 
if a random vector $(X_1,\dots,X_n)$ is pairwise counter-monotonic of type \eqref{eq:PCT2}, then 
 it is the maximizer of a risk sharing problem for some quantile agents.

 Theorem \ref{th:quantile} (i) assumes that $S$ is bounded from below. This is needed because any type-\eqref{eq:PCT1} pairwise counter-monotonic allocation is bounded from below. 
Theorem \ref{th:quantile} (ii) assumes $\p(S=\essinf S)>0$. 
 In case $\p(S>\essinf S)=0$, 
a pairwise counter-monotonic allocation of type \eqref{eq:PCT1} may not be Pareto optimal    for  any quantile agents with levels  in $(0,1)$.
A counter-example is provided in Example \ref{ex:1} below.
Theorem \ref{th:quantile} (iii) assumes that $S$ is continuously distributed. 
This condition is also needed for the result to hold. For instance, if $S=1$, then the allocation $(1/n,\dots,1/n)$ is Pareto optimal for any compatible quantile agents, violating the impossibility statement on Pareto optimality.
\begin{example} \label{ex:1}
Suppose that $S$ is uniformly distributed on $[0,1]$, 
and $X_i=S\id_{A_i}$ for $(A_1,\dots,A_n)\in \Pi_n$ independent of $S$ with $\p(A_i)>0$ for each $i\in [n]$. 
We will see that the pairwise counter-monotonic allocation $(X_1,\dots,X_n)$ is not Pareto optimal for any quantile agents with levels $\alpha_1,\dots,\alpha_n \in (0,1)$. 
If $\sum_{i=1}^n \alpha_i\ge 1$, there does not exist any Pareto-optimal allocation.
If $\sum_{i=1}^n\alpha_i<1$, then
$$
\sum_{i=1}^n \VaR_{\alpha_i} (X_i) = \sum_{i=1}^n \left(1 -\frac{\alpha_i}{\p(A_i)}\right)_+= \sum_{i=1}^n \left(\frac{\p(A_i)-\alpha_i}{\p(A_i)}\right)_+
$$
and 
$$\VaR_{\sum_{i=1}^n \alpha_i} (S)=1-\sum_{i=1}^n \alpha_i= \sum_{i=1}^n (\p(A_i)-\alpha_i)\le \sum_{i=1}^n \left(\frac{\p(A_i)-\alpha_i}{\p(A_i)}\right)_+= \sum_{i=1}^n \VaR_{\alpha_i} (X_i).$$
Using the condition \eqref{eq:optimality}, if $(X_1,\dots,X_n)$ is Pareto optimal, then the inequality above is an equality; this implies 
 $\alpha_i=\p(A_i)$ for each $i\in [n]$. However, this further implies $ \sum_{i=1}^n \alpha_i = \sum_{i=1}^n\p(A_i)=1$ conflicting 
$\sum_{i=1}^n \alpha_i<1$. 
 \end{example}
 
 The next example illustrates that for the same $S$ in Example \ref{ex:1} and compatible quantile agents,
a pairwise counter-monotonic Pareto-optimal allocation
exists as implied by Theorem \ref{th:quantile} (i). 
 
 \begin{example}
 Let  $S$ be uniformly distributed on $[0,1]$ and $\alpha_1, \dots, \alpha_n \in (0,1)$ with $\sum_{i=1}^n\alpha_i<1$.
 Take $(A_1,\dots,A_{n})\in \Pi_n$ such that $\bigcup_{i=1}^{n-1}A_i=\{S\ge 1-\sum_{i=1}^{n-1} \alpha_i\}$ and $\p(A_i)=\alpha_i$ for $i\in [n-1]$.
  Let 
  $X_i=S\id_{A_i}$  for $i\in [n]$.
 We can verify that $\VaR_{\alpha_i}(X_i)=0$ for $i\in [n-1]$ and
 $$\VaR_{\alpha_n}(X_n) = 
 \VaR_{\alpha_n}\left (S\id_{\{S<1-\sum_{i=1}^{n-1} \alpha_i\}}\right ) 
 = 1-\sum_{i=1}^{n} \alpha_i=\VaR_{\sum_{i=1}^n \alpha_i} (S). $$
This shows that $(X_1,\dots,X_n)$ is Pareto optimal.
It is also pairwise counter-monotonic by Theorem \ref{th:PCT}.
Note that although the allocation $(X_1,\dots,X_n)$ here has  the same form 
$(S\id_{A_1},\dots,S\id_{A_n})$
as the one in Example \ref{ex:1}, the specification of $(A_1,\dots,A_n)$ is different in the two examples, leading to opposite conclusions on optimality. 
\end{example}

\begin{remark}
One may notice that the condition on $S$ in Theorem \ref{th:quantile} part (ii) and that in part (iii), although both quite weak,
are actually conflicting. 
This is not a coincidence, because comonotonicity and counter-monotonicity have a non-empty intersection: 
A random vector is both comonotonic and counter-monotonic if and only if it has at most one non-degenerate component.
Therefore, we cannot have both conclusions in  parts (ii) and (iii) for the same $S$.
\end{remark}

\begin{remark}
As shown by \cite{ELW18}, the same pairwise counter-monotonic allocation which is Pareto optimal for quantile agents is also optimal for the more general Range Value-at-Risk (RVaR) agents. Therefore, the conclusion in Theorem \ref{th:quantile} also applies to the RVaR agents. 
Another appearance of  pairwise counter-monotonicity in optimal allocations  
is obtained by \cite{LLW23}, where it is shown that for agents using inter-quantile differences, a Pareto-optimal allocation is the sum of two  pairwise counter-monotonic random vectors. 
All   discussions above assume homogeneous beliefs; that is, all agents use the same probability measure $\p$.
In the setting of heterogeneous beliefs, \cite{ELMW20} showed  that for   Expected Shortfall  agents, 
a Pareto-optimal allocation above certain constant level 
also has a pairwise counter-monotonic structure; see their Proposition 3.
Generally, agents using the dual utility model of \cite{Y87}, including the quantile-based models above,  
have quite different features in risk sharing and other optimization problems compared to those with expected utility agents. 
For the optimal payoff of Yaari agents in portfolio choice, see 
\cite{BDV22}.
\end{remark}

 \begin{example}\label{ex:auction}
{ We illustrate that counter-monotonicity may also be the structure of an optimal  allocation outside the dual utility of \cite{Y87}.
  Let  $(\Omega,\mathcal A,\p)$ be an atomless probability space, $S=1$ and $\alpha>0$. Consider the problem 
\begin{align*}
\mbox{to maximize~~~}  \sum_{i=1}^n \E\left [\alpha \id_{\{X_i \geq 1\}}\right] ~~~\mbox{subject to}~(X_1,\dots,X_n)\in \mathbb A_n(S) ~\mbox{and}~X_i\geq 0, \mbox{ for } i\in [n].
\end{align*}
It is straightforward  to verify that the set of maximizers is 
$$\mathbb A^*=\left\{(\id_{A_1},\dots,\id_{A_n})\in \mathbb A_n(S): (A_1,\dots,A_n)\in \Pi_n\right\},$$
which contains only counter-monotonic allocations.
This problem can be interpreted as  the problem of sharing $S=1$ among $n$ expected utility maximizers with common   utility function $u(x)=\alpha \id_{\{x \geq 1\}}$ for $\alpha >0$. 
The optimization problem is thus a social planner's problem, and the set $\mathbb A^*$ contains all Pareto-optimal allocations for this problem. 
The allocations satisfying $\p(A_i)=\p(A_j)$ for every $i\neq j$ are of particular interest, as they are common in auction theory as the \textit{random tie-breaking rule}. The variable $S$ can be understood as an indivisible good that was auctioned, and the parameter $\alpha$ as the net utility of a series of $n$ agents with quasi-linear utilities $v(X, t)= \theta X - t$ having bid the same amount $0\leq t <\theta$. It is straightforward to see that these allocations are the only \textit{fair} allocations, in the sense that all agents have the same expected utility. In other words, a fair lottery (which is counter-monotonic) is the only fair way to distribute the indivisible good among people who value it equally.}
 \end{example}

\section{Conclusion}

We provide a series of 
technical results on the representation (Theorem \ref{th:PCT}) and invariance property (Theorem \ref{th:CT}) of pairwise counter-monotonicity, as well as their connection to 
negative association (Theorem \ref{th:NA}), joint mix dependence (Theorem \ref{th:frechet}), and optimal allocations for quantile agents (Theorem \ref{th:quantile}). 
Our paper is motivated by the recently increasing attention in counter-monotonicity and negative dependence, and it fills the gap between the relatively scarce studies on pairwise counter-monotonicity  in the literature
and the wide appearance of this dependence structure in modern applications,
in particular, in risk sharing problems with agents that are not using expected utilities.

In general, studies of 
negative dependence
and positive dependence are highly asymmetric in nature, 
with negative dependence being 
 more challenging to study  in various applications of risk management and statistics. 
In addition to the negative dependence  concepts  we considered in this paper, 
some other notions have been studied in the recent literature, and the interested reader is referred to \cite{ASB13}, \cite{LA14}, \cite{LCA17} and \cite{CRW22}, as well as the monographs of \cite{J97,J14}.


\bigskip
 \textbf{Acknowledgements.} We thank the handling editor and two anonymous referees for helpful comments on a previous version of the paper. 
 Ruodu Wang acknowledges support from the Canada Research Chairs (CRC-2022-00141) and the Natural Sciences and Engineering Research Council of Canada (NSERC Grant No.~RGPIN-2018-03823).

\end{document}